\documentclass[conference,a4paper]{IEEEtran}

\usepackage{tikz}
\usetikzlibrary{matrix}
\usepackage{graphicx}
\usepackage{epstopdf}
\usepackage{wrapfig}
\usepackage{amsfonts}
\usepackage{makecell}
\usepackage{pdfpages}
\usepackage[utf8]{inputenc} 
\usepackage{url}
\usepackage{ifthen}
\usepackage{cite}
\usepackage[cmex10]{amsmath}
\usepackage{mathtools}
\usepackage{amsthm}
\usetikzlibrary{shapes,arrows}
\usepackage{amssymb}

\usepackage[T1]{fontenc} 
\usepackage{mleftright}       
\mleftright                   

\usepackage{graphicx}         
\usepackage{booktabs}         

\usepackage{stfloats}
\usepackage{float}

\newtheorem{theorem}{Theorem}

\DeclareMathOperator*{\argmin}{arg\,min}

\begin{document}

\title{Efficient Evaluation of the Probability of Error of Random Coding Ensembles} 



\author{%
  \IEEEauthorblockN{Ioannis Papoutsidakis, Angela Doufexi, and Robert J. Piechocki}\\
  \IEEEauthorblockA{Communication Systems and Networks Group\\ 
  					Department of Electrical and Electronic Engineering\\
                    University of Bristol\\ 
                    Bristol, BS8 1UB, UK\\
                    Email: \{ioannis.papoutsidakis, a.doufexi, r.j.piechocki\}@bristol.ac.uk}
}

\maketitle

\begin{abstract}
This paper presents an achievability bound that evaluates the exact probability of error of an ensemble of random codes that are decoded by a minimum distance decoder. Compared to the state-of-the-art which demands exponential computation time, this bound is evaluated in polynomial time. This improvement in complexity is also attainable for the original random coding bound that utilizes an information density decoder. The general bound is particularized for the binary symmetric channel, the binary erasure channel, and the Gaussian channel.
\end{abstract}

\section{Introduction}

The error rate of the optimal code of blocklength $n$ is an interesting subject of information theory for many decades. It is a practically important problem because it captures the tradeoff between the rate and the delay and complexity of a communication system.

One of the first works that address it is Shannon's for the Gaussian channel, where lower and upper bounds on the probability of error of the optimal code are introduced based on random coding and sphere packing \cite{6767457}. These results are numerically evaluated and studied in \cite{slep} for $n$ up to $100$. More recently, the work of Polyanskiy et al. \cite{poly} introduced several achievability bounds based on a maximum likelihood decoder that maximizes the information density metric. Specifically, for discrete memoryless channels, they derive the random coding (RC) bound which evaluates the exact probability of error of an ensemble of random codes, the random coding union (RCU) bound which is an efficient relaxation of RC bound, and the dependence testing (DT) bound which relates to binary hypothesis testing.

The RC bound evaluates the exact probability of error of an ensemble of random codes. As a result, it is expected to be tighter than RCU and DT bounds \cite{6620522}. Despite this fact, it is not generally preferred due to its exponential complexity. The efficient calculation of RC bound is a very interesting open problem.

The utilization of information density is prevalent in the aforementioned results as well as in many finite-blocklength results in the literature. Nevertheless, for several important channels such as the binary symmetric channel (BSC), binary erasure channel (BEC), and the Gaussian channel it is well known that minimum distance decoding is equivalent to maximum likelihood decoding. The relation of these metrics becomes apparent in the cases of BSC and BEC, since information density is a function of minimum distance \cite{poly}. 
An achievability bound that is based on minimum distance decoding is of interest because it allows the utilization of well-known results from probability theory, especially for the Gaussian channel.

The current paper provides an efficient RC bound that is based on minimum distance decoding and can be evaluated in polynomial time. Specifically, we provide an alternative form that avoids a sum with exponentially many terms and can be also used with the information density metric. Furthermore, for the case where the distance metric is continuous, we show how the RC bound is simplified. We particularize the general result for BSC and BEC and discuss how to efficiently evaluate them by dealing with specific computational challenges. Finally, the mathematical expression of the bound is given for the Gaussian channel with a constraint on average power.

In section \ref{notation}, the original RC bound is presented as well as the notation we follow throughout this paper. The main results are given in section \ref{results}. The particularization of the main result for the BSC, BEC, and the Gaussian channel is given in sections \ref{binary} and \ref{gaussian}, respectively. The final remarks and conclusions are made in section \ref{concl}.

\section{Background and Notation}
\label{notation}
Let us consider a channel with input alphabet $A$ and output alphabet $B$ with a conditional probability $P_{Y|X}:A\mapsto B$. An arbitrary codebook for this channel is denoted as $(c_1,...c_M)\in A^M$ where $M$ is the codebook size. The information density for a joint distribution $P_{XY}$ on $A \times B$ is
\begin{align}
i(x;y) = \log \frac{dP_{Y|X=x}}{dP_Y}(y).
\end{align}
We give in this section the original RC bound since the main results are derived based on its modification.
\begin{theorem}
Denote by $\epsilon(c_1,...,c_M)$ the error probability achieved by the maximum likelihood decoder with codebook $(c_1,...,c_M)$. Let $X_1,...,X_M$ be independent with marginal distribution $P_X$. Then 
\begin{align}
\mathbb{E}[\epsilon(X_1,...,X_M)]=1-\sum_{l=0}^{M-1} \binom{M-1}{l} \frac{1}{1+l}\mathbb{E}\big[w^lz^{M-1-l}\big]
\end{align}
where
\begin{align}
w=P(i(\bar{X};Y)=i(X;Y)|X,Y)
\end{align}
\begin{align}
z=P(i(\bar{X};Y)<i(X;Y)|X,Y)
\end{align}
with
\begin{align}
P_{XY\bar{X}}(a,b,c)=P_X(a)P_{Y|X}(b|a)P_X(c).
\end{align}
\end{theorem}
Observe that this bound requires the evaluation of a sum with $M$ terms where $M$ grows exponentially with blocklength $n$ and rate $R$ measured in bits/channel use,  since $M=2^{nR}$. As a result, its evaluation is difficult unless $M$ is small enough. The high complexity of RC bound makes the relaxations provided by RCU and DT bounds practical.

Throughout this paper, blocklength is denoted with $n$. The multivariate normal distribution is denoted with $\mathcal{N}(\mu,\Sigma)$, where $\mu$ is the mean vector and $\Sigma$ is the covariance matrix. $\textbf{I}_m$ stands for the $m \times m$ identity matrix. The probability density function (pdf) of the gamma distribution is denoted with $f_\Gamma(x;\kappa,\theta)$, where $\kappa$ is the shape parameter and $\theta$ is the scale parameter. The pdf of the non-central chi-squared distribution with degrees of freedom $\kappa$ and non-centrality parameter $\lambda$ is denoted with $f_{X^2}(x;\kappa,\lambda)$. The cumulative distribution functions (cdf) are denoted analogously but with an upper case function name, e.g. $F_{X^2}(x;\kappa,\lambda)$.

\section{Random coding bound}
Random coding is a widely used tool in coding and information theory. It is also the main tool we use to derive the main results.
\label{results}
\begin{theorem}
Denote by $\epsilon(c_1,...,c_M)$ the error probability achieved by the minimum distance decoder with codebook $(c_1,...,c_M)$. Let $X_1,...,X_M$ be independent with marginal distribution $P_X$. Then 
\begin{align}
\mathbb{E}[\epsilon(X_1,...,X_M)]=1-\mathbb{E}\bigg[\frac{(w+z)^M-z^M}{wM}\bigg]
\label{boundrc}
\end{align}
where
\begin{align}
w=P(d(\bar{X},Y)=d(X,Y)|X,Y)
\end{align}
\begin{align}
z=P(d(\bar{X},Y)>d(X,Y)|X,Y)
\end{align}
with
\begin{align}
P_{XY\bar{X}}(a,b,c)=P_X(a)P_{Y|X}(b|a)P_X(c).
\end{align}
\label{rc}
\end{theorem}

\begin{proof}
Initially, the proof follows similar steps as the proof of \cite[Theorem 15]{poly}. Upon reception of channel output $y$ and given the codebook is $(c_1,...,c_M)$, the minimum distance decoder estimates the transmitted message
\begin{align} 
\hat{m}=\argmin_{i=1,...,M} d(c_i,y).
\end{align}
Assume, without loss of generality, that $m=1$ and the corresponding codeword is $c_1$. This estimation is correct with probability $\frac{1}{1+l}$ if
\begin{align}
\sum_{j=2}^M 1\{d(c_j,y)=d(c_1,y) \}=l \text{ and}
\end{align}
\begin{align}
\sum_{j=2}^M 1\{d(c_j,y)<d(c_1,y) \}=0
\label{ertra}
\end{align}
for $l=0,...,M-1$. If (\ref{ertra}) is not satisfied then an error occurs with absolute certainty. Let,
\begin{align}
w=P(d(\bar{X},y)=d(c_1,y)) \text{ and}
\end{align}
\begin{align}
z=P(d(\bar{X},y)>d(c_1,y))
\end{align}
where $\bar{X}$ is an arbitrary codeword other than $c_1$ and $y$ is the channel output. Since the codewords are independent and identically distributed the joint distribution of the remaining codewords is $P_X \times ...\times P_X$.Therefore, the conditional probability of correct decision is,

\begin{align}
\begin{split}
P(\hat{m}=1|y)&= \sum_{l=0}^{M-1} \binom{M-1}{l} \frac{1}{1+l}w^lz^{M-1-l}\\
&\stackrel{(a)}{=}\sum_{l=0}^{M-1} \binom{M}{l+1} \frac{1}{M}w^lz^{M-1-l}\\
&= \frac{1}{wM}\sum_{l=0}^{M-1} \binom{M}{l+1} w^{l-1}z^{M-1-l}\\
&\stackrel{(b)}{=} \frac{1}{wM}\sum_{k=1}^{M} \binom{M}{k} w^{k}z^{M-k}\\
&\stackrel{(c)}{=}\frac{(w+z)^M-z^M}{wM}
\end{split}
\label{condprob}
\end{align}
where $(a)$ comes from the absorption identity of binomial coefficients \cite{10.5555/562056}, $(b)$ comes from change of variables, and $(c)$ comes from the binomial theorem. Averaging (\ref{condprob}) with respect to $(c_1,y)$ jointly distributed as $P_{XY}$ we obtain equation (\ref{boundrc}).

\end{proof}

Theorem \ref{rc} assumes an arbitrary memoryless channel. However, in the case which the channel is continuous and the resulting distances are continuous random variables the bound simplifies as follows.

\begin{theorem}
Denote by $\epsilon(c_1,...,c_M)$ the error probability achieved by the minimum distance decoder with codebook $(c_1,...,c_M)$. Let $X_1,...,X_M$ be independent with marginal continuous distribution $P_X$. Then 
\begin{align}
\mathbb{E}[\epsilon(X_1,...,X_M)]=1-\mathbb{E}[z^{M-1}]
\end{align}
where
\begin{align}
z=P[d(\bar{X},Y)>d(X,Y)|X,Y]
\end{align}
with
\begin{align}
P_{XY\bar{X}}(a,b,c)=P_X(a)P_{Y|X}(b|a)P_X(c).
\end{align}
\label{rc2}
\end{theorem}

\begin{proof}
We apply Theorem \ref{rc}. Since $d(\bar{X},Y)$ and $d(X,Y)$ are continuous random variables we have
\begin{align}
w=P[d(\bar{X},Y)=d(X,Y)|X,Y]=0.
\end{align}
Using L'Hospital's rule we have the following
\begin{align}
\lim_{w \rightarrow 0}\frac{(w+z)^M-z^M}{wM} = \lim_{w \rightarrow 0}\frac{M(w+z)^{M-1}}{M}=z^{M-1}
\end{align}
\end{proof}

This result is very useful because it simplifies the bound for continuous channels without relaxing it. The potential of this simplification is also referred to in \cite{6620522} and utilized in \cite{9328341}. We include it in this paper for completeness.

%
%
%

\section{Binary Discrete Channels}
\label{binary}
This section particularizes Theorem \ref{rc} for two important binary memoryless discrete channels, namely the binary symmetric channel and the binary erasure channel. The resulting achievability bounds are the same as the bounds of \cite[Theorem 32, Theorem 36]{poly} since information density is a function of minimum distance in these specific channels. However, the important difference is that our bounds can be computed in polynomial time in contrast to the original ones that have exponential time complexity.

\begin{theorem}
For the BSC with error probability $\delta$, we have
\begin{align}
\begin{split}
\mathbb{E}&[\epsilon(X_1,...,X_M)]\\
&=1-\sum_{i=0}^n \delta^i(1-\delta)^{n-i}\frac{\big(\sum_{j=i}^n\binom{n}{j}\big)^M-\big(\sum_{j=i+1}^n\binom{n}{j}\big)^M}{M 2^{nM-n}}
\end{split}
\end{align}
\end{theorem} 
\begin{proof}
The appropriate distance metric for the BSC is the Hamming distance
\begin{align}
d(X,Y) = \sum_{i=1}^n X_i \oplus Y_i 
\end{align}
where $\oplus$ denotes the addition over $GF(2)$.

Note that since the codebook is random and each symbol  follows the Bernoulli$(0.5)$, the resulting $M$ distances are independent. Specifically, $d(X,Y)$ follows the Binomial$(n,\delta)$ and the rest $M-1$ distances follow the Binomial$(n,0.5)$. The final derivation of $\mathbb{E}[\epsilon(X_1,...,X_M)]$ can be found in (\ref{bsc25}) at the bottom of the next page.
\begin{figure*}[b]
\rule[1ex]{\textwidth}{0.1pt}
\begin{align}
\begin{split}
\mathbb{E}[\epsilon(X_1,...,X_M)]&=1-\sum_{i=0}^n \binom{n}{i}\delta^i(1-\delta)^{n-i}\frac{\big(\binom{n}{i}2^{-n}+\sum_{j=i+1}^n\binom{n}{j}2^{-n}\big)^M-\big(\sum_{j=i+1}^n\binom{n}{j}2^{-n}\big)^M}{\binom{n}{i}2^{-n}M}\\
&=1-\sum_{i=0}^n \binom{n}{i}\delta^i(1-\delta)^{n-i}\frac{\big(\sum_{j=i}^n\binom{n}{j}2^{-n}\big)^M-\big(\sum_{j=i+1}^n\binom{n}{j}2^{-n}\big)^M}{\binom{n}{i}2^{-n}M}\\
&=1-\sum_{i=0}^n \delta^i(1-\delta)^{n-i}\frac{\big(\sum_{j=i}^n\binom{n}{j}\big)^M-\big(\sum_{j=i+1}^n\binom{n}{j}\big)^M}{M 2^{nM-n}}
\end{split}
\label{bsc25}
\end{align}
\end{figure*}
\end{proof}

\begin{theorem}
For the BEC with erasure probability $\delta$, we have
\begin{align}
\begin{split}
\mathbb{E}[\epsilon(X_1&,...,X_M)]\\
&=1-\sum_{i=0}^n \binom{n}{i}\delta^i(1-\delta)^{n-i}\frac{1-(1-2^{i-n})^M}{2^{i-n}M}
\end{split}
\end{align}
\label{bec}
\end{theorem} 
\begin{proof}
Similarly to BSC, the appropriate metric is the Hamming distance. The distance of the correct codeword is equal to the number of erased symbols $k$,
\begin{align}
d(X,Y)=k.
\end{align}
The distance of the rest of the codewords is also a function of $k$,
\begin{align}
d(\bar{X},Y)=k + u
\end{align}
where $k\sim$ Binomial$(n,\delta)$ and $u\sim$ Binomial$(n-k,0.5)$. The final derivation of $\mathbb{E}[\epsilon(X_1,...,X_M)]$ can be found in (\ref{bec25}) at the bottom of the next page.

\begin{figure*}[b]
\rule[1ex]{\textwidth}{0.1pt}
\begin{align}
\begin{split}
\mathbb{E}[\epsilon(X_1,...,X_M)]&=1-\sum_{i=0}^n \binom{n}{i}\delta^i(1-\delta)^{n-i}\frac{\big(2^{-(n-i)}+\sum_{j=1}^{n-i}\binom{n-i}{j}2^{-(n-i)}\big)^M-\big(\sum_{j=1}^{n-i}\binom{n-i}{j}2^{-(n-i)}\big)^M}{2^{-(n-i)}M}\\
&=1-\sum_{i=0}^n \binom{n}{i}\delta^i(1-\delta)^{n-i}\frac{1-\big(\sum_{j=1}^{n-i}\binom{n-i}{j}2^{-(n-i)}\big)^M}{2^{-(n-i)}M}\\
&=1-\sum_{i=0}^n \binom{n}{i}\delta^i(1-\delta)^{n-i}\frac{1-(1-2^{i-n})^M}{2^{i-n}M}
\end{split}
\label{bec25}
\end{align}
\end{figure*}
\end{proof}

The derived bounds can be computed in polynomial time, however the arithmetic underflows causes by the exponentiation with extremely large exponents is a challenge. This is easily resolved by representing the terms of the sums on a logarithmic scale and converting the operations appropriately. As a proof of concept, we give the evaluation of Theorem \ref{bec} in Figure \ref{f1}. As expected, it outperforms the rest state-of-the-art achievability bounds.

\begin{figure}
\centering
\includegraphics[scale=0.57]{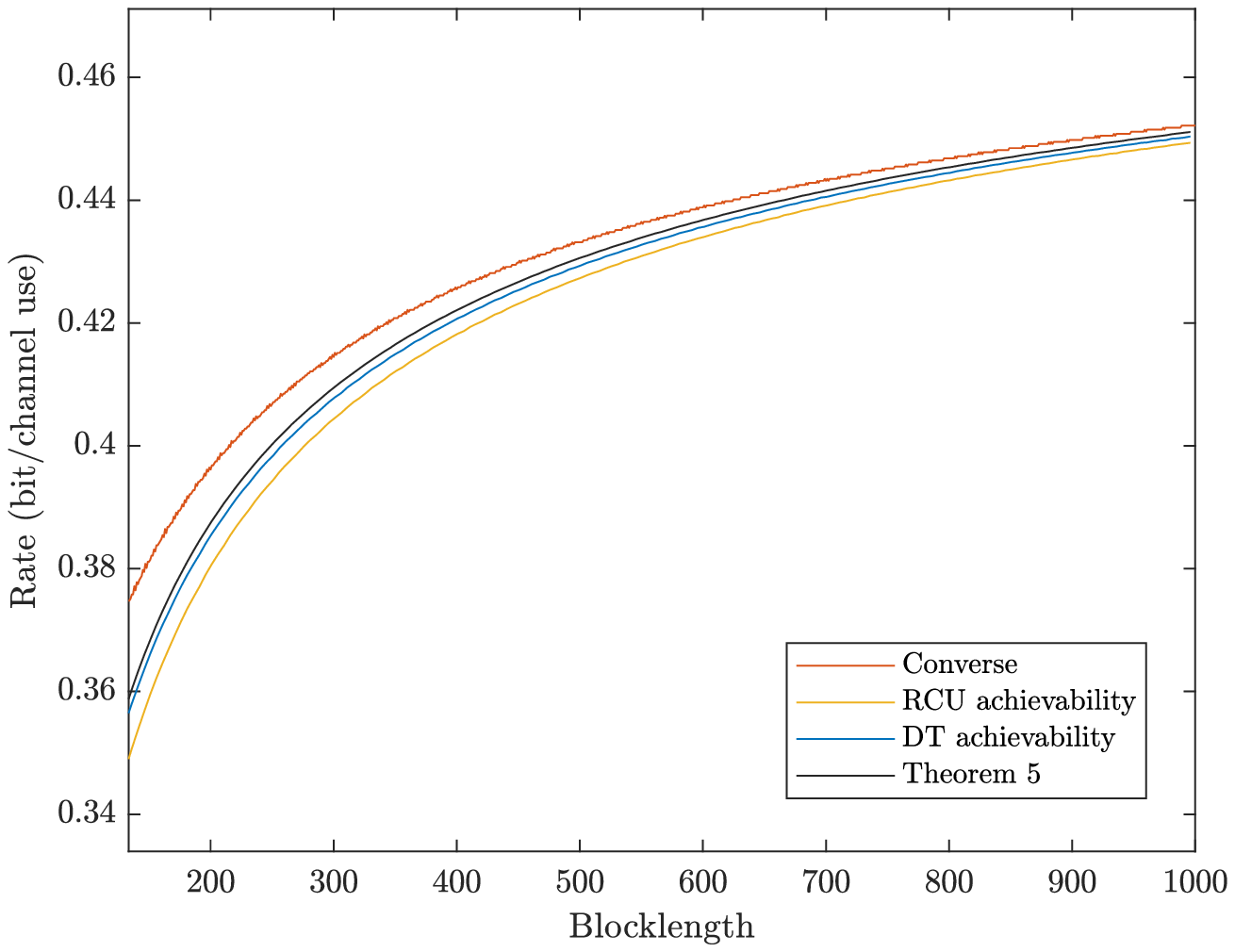}
\caption{Comparison of Theorem \ref{bec} with RCU and DT bounds for the BEC with erasure rate $\delta=0.5$ and average error probability $\epsilon = 10^{-3}$. The converse bound is the one in \cite[Theorem 38]{poly}}.
\label{f1}
\end{figure}

\section{The Gaussian Channel}
\label{gaussian}
The most fundamental and well studied continuous channel is the Gaussian channel. Specifically, it is a real-valued channel with additive Gaussian noise and a power constraint. In this section, we apply Theorem \ref{rc2} and derive its mathematical expressions for this channel.

For random coding over the Gaussian Channel, it is convenient to use the following definitions.
\begin{itemize}
\item Channel input $X_1,...X_n \sim \mathcal{N}(\boldsymbol 0,\textbf{I}_n)$,
\item Noise $Z_1,...Z_n \sim \mathcal{N}(\boldsymbol 0,\sigma_Z^2\textbf{I}_n)$,
\item Channel output $Y_1=X_1+Z_1,...,Y_n=X_n+Z_n$,
\item Signal-to-noise ratio $\gamma=1/\sigma_Z^2$.
\end{itemize} 
Note that the power constraint is the average over the ensemble of random codes that are produced with this process,
\begin{align}
\mathbb{E}\Bigg(\frac{1}{M}\sum_{i=1}^M \Vert c_i\Vert^2 \Bigg) = n.
\end{align}
This constraint is different from the one defined in \cite{6767457} for average codeword power. Although, it is a good approximation due to the law of large numbers since $M$ grows exponentially with $nR$. Thus, for sufficiently large $M$,
\begin{align}
\frac{1}{M}\sum_{i=1}^M \Vert c_i\Vert^2 \approx n.
\end{align}

\begin{theorem}
For the Gaussian channel with signal-to-noise ratio $\gamma$,
\begin{align}
\begin{split}
\mathbb{E}&[\epsilon(X_1,...,X_M)] \\
&= \int_0^{\infty}\int_0^{\infty}f_\Gamma(x;2^{-1}n,2\gamma^{-1})f_{X^2}(y;n,x)\\
&\qquad\qquad\qquad(1-(1-F_{X^2}(x;n,y))^{M-1})\,dx\,dy.
\end{split}
\end{align}
\label{gaus11}
\end{theorem}

\begin{proof}
The input, output, and noise vectors of the Guassian channel exist in the $n$-dimensional real space $\mathbb{R}^n$. Therefore, the appropriate distance metric is the Euclidean distance,
\begin{align}
d_E(X,Y) = \sqrt{(X_1-Y_1)^2+...+(X_n-Y_n)^2}.
\end{align}

Since the distances are just compared, it is convenient to use the squared Euclidean distance because many distributions that describe it are readily available and the result  remains the same,
\begin{align}
d(X,Y) = (X_1-Y_1)^2+...+(X_n-Y_n)^2.
\end{align}

The codewords are independent, identically distributed, and follow the multivariate Gaussian distribution $\mathcal{N}(\boldsymbol 0,\textbf{I}_n)$. The noise vector has a multivariate Gaussian distribution $\mathcal{N}(\boldsymbol 0,\gamma^{-1}\textbf{I}_n)$. Again, we use standard normal random codewords for convenience without loss of generality.

We apply Theorem \ref{rc2}. Note that $\lambda_Z \triangleq d(X,Y)$ follows the scaled chi-squared distribution or equivalently the gamma distribution with shape parameter $2^{-1}n$ and scale parameter $2\gamma^{-1}$. Additionally, $J_{\lambda_Y}\triangleq d(\bar{X},Y)$ follows the non-central chi-squared distribution with non-centrality parameter
\begin{align}
\lambda_Y = \sum_{i=1}^{n}Y_i^2.
\end{align}
Lastly, $\lambda_Y$ follows the non-central chi-squared distribution with non-centrality parameter
\begin{align}
\lambda_Z = \sum_{i=1}^{n}Z_i^2.
\end{align}
Therefore,
\begin{align}
\begin{split}
\mathbb{E}&[\epsilon(X_1,...,X_M)]\\
&=1-\mathbb{E}[(P[d(\bar{X},Y)>d(X,Y)|X,Y])^{M-1}] \\
&=\mathbb{E}[1-(1 - P[d(\bar{X},Y)\leq d(X,Y)|X,Y])^{M-1}]\\
&=\mathbb{E}[1-(1 - P[J_{\lambda_Y}\leq \lambda_Z|\lambda_Y,\lambda_Z])^{M-1}]\\
&= \int_0^{\infty}\int_0^{\infty}f_\Gamma(x;2^{-1}n,2\gamma^{-1})f_{X^2}(y;n,x)\\
&\qquad\qquad\qquad(1-(1-F_{X^2}(x;n,y))^{M-1})\,dx\,dy.
\end{split}
\end{align}
\end{proof}

An equivalent evaluation of the average probability of error of this random coding ensemble can be found in \cite{9328341}, where the author utilizes the radial and tangential components of the received codeword.

Again, evaluating numerically this integration is challenging due to the exponentiation. Even though it is possible to approximate it using operations with logarithmic probabilities, it is easier to derive two simple lower and upper bounds and compare their tightness.

\begin{theorem}
For the Gaussian channel with signal-to-noise ratio $\gamma$,
\begin{align}
\begin{split}
\mathbb{E}[\epsilon(X_1&,...,X_M)]\\
&\leq \int_0^{\infty}\int_0^{\infty}f_\Gamma(x;2^{-1}n,2\gamma^{-1})f_{X^2}(y;n,x)\\
&\qquad\qquad\quad\min\{1, (M-1)F_{X^2}(x;n,y)\}\,dx\,dy,
\end{split}
\label{rcuG}
\end{align}

\begin{align}
\begin{split}
\mathbb{E}&[\epsilon(X_1,...,X_M)]\\
&\geq \int_0^{\infty}\int_0^{\infty}f_\Gamma(x;2^{-1}n,2\gamma^{-1})f_{X^2}(y;n,x)\\
&\qquad\bigg(1-\frac{1}{1-(M-1)\log(1-F_{X^2}(x;n,y))}\bigg)\,dx\,dy.
\end{split}
\label{upperG}
\end{align}
\end{theorem}

\begin{proof}
Inequality (\ref{rcuG}) comes from Bernoulli's inequality and the fact that $1-(1-F_{X^2}(x;n,y))^{M-1} \leq 1$. For inequality (\ref{upperG}), let
\begin{align}
a = -\log(1-F_{X^2}(x;n,y)).
\end{align}
Then
\begin{align}
\begin{split}
1-(1-F_{X^2}(x;n,y))^{M-1} &= 1-\frac{1}{e^{(M-1)a}}\\
&=1-\frac{1}{\sum_{k=0}^{\infty}\frac{((M-1)a)^k}{k!}}\\
&\geq1-\frac{1}{\sum_{k=0}^{1}\frac{((M-1)a)^k}{k!}}\\
&= 1-\frac{1}{1+(M-1)a}.
\end{split}
\label{b39}
\end{align}
\end{proof}
Bound (\ref{rcuG}) can be seen as a random coding union bound. In Figure \ref{f2}, bounds (\ref{rcuG}) and (\ref{upperG}) are plotted for signal-to-noise ratio $\gamma=1$ (0 dB) and average error probability $\epsilon = 10^{-3}$. For this setting of parameters, bounds (\ref{rcuG}) and (\ref{upperG}) are very close. Evaluation of more terms of the sum in (\ref{b39}) can provide excellent approximations of Theorem \ref{gaus11}.

\begin{figure}
\centering
\includegraphics[scale=0.57]{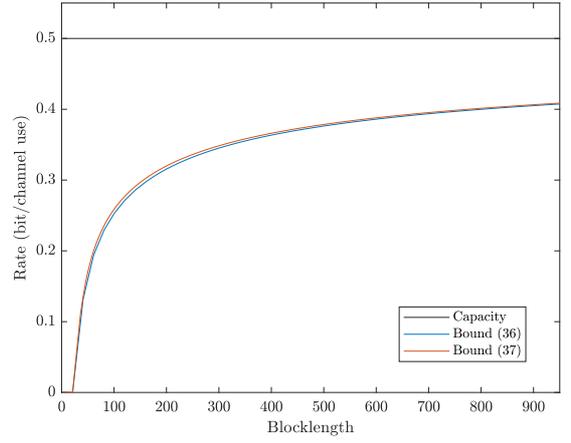}
\caption{Comparison of bounds (\ref{rcuG}) and (\ref{upperG}) for the Gaussian channel with SNR $\gamma=1$ (0 dB) and average error probability $\epsilon = 10^{-3}$.} 
\label{f2}
\end{figure}

\section{Conclusions}
\label{concl}
This work presents an achievability bound that evaluates the exact probability of error of an ensemble of random codes that are decoded by a minimum distance decoder. Compared to the state-of-the-art which demands exponential computation time, this bound is evaluated in polynomial time. This improvement in complexity is also attainable for the original bound that utilizes an information density decoder. The general bound is applied for the BSC, BEC, and the Gaussian channel. The numerical evaluation for the BEC verifies the higher achievable rate compared to the relaxations of RCU and DT bounds. For the Gaussian channel, upper and lower bounds to the exact probability of error are derived. These bounds are very close in the presented setting. The rationale of the minimum distance as a decoding metric can be valuable for other applications, such as variable length coding with feedback, especially for the Gaussian channel.

\section*{Acknowledgment}
This work is supported by the Engineering and Physical Sciences Research Council (EP/L016656/1) and the University of Bristol.

\bibliography{rc_min_distance}

\bibliographystyle{ieeetr}

\end{document}